\documentclass[10pt]{article}
\usepackage{amsmath}
\usepackage{amssymb, amscd, amsthm}
\usepackage[all]{xy}
\usepackage[dvips]{graphicx}
\usepackage{verbatim}
\usepackage[perpage,symbol*]{footmisc}
\newtheorem{theorem}{Theorem}

\newtheorem{lemma}{\textbf{Lemma}}[section]
\newtheorem{remark}{\textbf{Remark}}[section]
\newtheorem{corollary}{\textbf{Corollary}}[section]

\newtheorem{example}{\textbf{Example}}[section]

\usepackage{longtable}
\usepackage{multirow}
\usepackage{slashbox}

\newcommand{\F}{\mathbb{F}}

\begin{document}

\baselineskip 17pt
\title{\Large\bf Construction of MDS Self-dual Codes over Finite Fields}

\author{\large  Khawla Labad\qquad Hongwei Liu\qquad Jinquan
Luo*}\footnotetext{The authors are with school of mathematics and
statistics, Central China Normal University, Wuhan, China. \\
E-mail: 2768400154@qq.com(K.Labad),\quad
hwliu@mail.ccnu.edu.cn(H.Liu),\quad
luojinquan@mail.ccnu.edu.cn(J.Luo)}
\date{}
\maketitle

{\bf Abstract}:
   In this paper, we obtain some new results on the existence of MDS self-dual codes utilizing (extended) generalized Reed-Solomon codes
  over finite fields of odd characteristic.  For some fixed $q$, our results can produce more classes
  of MDS self-dual  codes than previous works.

\section{Introduction}

 \quad\; A linear $[n,k,d]$ code over a finite field $\mathbb{F}_{q}$ of length $n$, dimension $k$ and minimum distance $d$ is called MDS
 (maximum distance separable) if it attains the Singleton bound: $d = n-k + 1$. MDS codes have been of much interest from many
 researchers due to their theoretical significant and practical implications, see [\ref{DSDY}], [\ref{KKO}], [\ref{LR}].

For a linear code $C$, we denote the dual of $C$ under Euclidean inner product by $C^{\bot}$.  The linear code $C$ is called self-dual if  $C = C^\bot$.
Self-dual codes have  attracted attention from  coding theory, cryptograph and other fields.
It has been found various applications in cryptography, in particular secret sharing scheme [\ref{CDG}], [\ref{DMS}], [\ref{MAS}] and combinatorics [\ref{MS}].

More recently, the application of MDS codes renewed the interest in
the construction of MDS self-dual codes, see [\ref{AKL}],
[\ref{BGG}], [\ref{GK}], [\ref{KL}]. K. Guenda [\ref{GUE}]
constructed MDS Euclidean and Hermitian self-dual codes which are
extended cyclic duadic codes or negacyclic codes. She also
constructed Euclidean self-dual codes which are extended negacyclic
codes.

Generalized Reed-Solomon ($\mathbf{GRS}$) codes is a class of MDS
code which has nice algebraic structure. It can be systematically
constructed and
 has been found wide applications in practice.
 MDS self-dual codes through $\mathbf{GRS}$ codes have been studied by L. Jin and C. Xing [\ref{JX}],
 where they constructed several classes of MDS self-dual codes through
$\mathbf{GRS}$ codes by choosing suitable parameters. In
[\ref{Yan}], H. Yan generalizes the technique in [\ref{JX}] and construct several classes
 of MDS self-dual codes via $\mathbf{GRS}$ codes and extended $\mathbf{GRS}$ codes.

Since MDS self-dual codes over finite field of even characteristic
with any possible parameter have been found in [\ref{GA}]. In this
paper, we obtain some new results on the existence of MDS self-dual
codes through (extended) $\mathbf{GRS}$ codes over finite fields of
odd characteristic.  Some results in this paper extend those of
[\ref{JX}] and [\ref{Yan}]. Comparing to previous works, for some
fixed square prime power $q$, our construction will produce more
$q$-ary MDS self-dual codes.

This paper is organized as follows. In Section 2 we will introduce
some basis knowledge and auxiliary results on $\mathbf{GRS}$ codes
and extended $\mathbf{GRS}$ codes. In particular, Corollary \ref{c1}
and Corollary \ref{c2} give a criterion for an (extended)
$\mathbf{GRS}$ code to be self-dual. In Section 3 we will present
our main results on the construction of MDS self-dual codes. Our
main tools are Corollary \ref{c1} and Corollary \ref{c2}. We choose
suitable parameters to make the conditions in Corollary \ref{c1} and
Corollary \ref{c2} hold.

\begin{center}
\begin{longtable}{|c|c|c|}  
\caption{Known results on MDS self-dual codes of length
$n$\qquad\qquad ( $\eta$ is the quadratic character of
$\mathbb{F}_{q}$) }\footnotetext{Some results in Table 1 are overlapped. For instance, columns 6 and 7 in [\ref{GUE}] are special cases of column 4 in [\ref{Yan}].} \\
\hline $q$ & $n$ even
& Reference\\  \hline $q$ even  &  $n \leq q$   & [\ref{GA}] \\
\hline $q$ odd & $n=q+1$ & [\ref{GA}]\\ \hline $q$ odd & $(n-1)|(q -
1)$, $\eta(1 - n) = 1$ &   [\ref{Yan}] \\ \hline $q$ odd & $
(n-2)|(q - 1)$, $\eta(2 - n) = 1$ &   [\ref{Yan}]\\ \hline
$q = r^{s}\equiv 3\pmod{4}$ &  $n-1= p^m \mid(q-1)$, prime $p\equiv 3\pmod{4}$ and odd $m$ &  [\ref{GUE}]\\
\hline
$q = r^{s}$, $r \equiv 1 \pmod{4}$, $s$ odd &  $n-1= p^m\mid (q-1)$, odd $m$ and prime $p \equiv 1\pmod{4}$  &  [\ref{GUE}]\\
\hline
$q = r^{s}$ , $r$ odd, $s\geq 2$ & $n = lr$, even $l$ and $2l|(r - 1)$ &   [\ref{Yan}] \\ \hline

$q = r^{s}$ , $r$ odd, $s \geq 2$ & $n = lr$, even $l$, $(l - 1)|(r - 1)$ and $\eta(1 - l)=1$ &   [\ref{Yan}] \\ \hline

$q = r^{s}$ , $r$ odd, $s \geq 2$ & $n = lr + 1$, odd $l$, $l|(r - 1)$ and $\eta(l) = 1$  &   [\ref{Yan}] \\ \hline
 $q = r^{s}$ , $r$ odd, $s \geq 2$ & $n = lr + 1$, odd $l$, $(l - 1)|(r - 1)$ and $\eta(l - 1) = \eta(-1) = 1$ &  [\ref{Yan}] \\ \hline

$q=r^2$  & $n \leq r$  & [\ref{JX}] \\ \hline
$q = r^2, r \equiv 3\pmod{4}$  &  $n= 2tr$ for any $t \leq (r - 1)/2$ &   [\ref{JX}]\\ \hline

$q = r^2$, $r$ odd & $n = tr$, even $t$ and $1 \leq t \leq r$ &   [\ref{Yan}] \\ \hline

 $q = r^2$, $r$ odd & $n = tr + 1$, odd $t$ and $1 \leq t \leq r$ &   [\ref{Yan}] \\ \hline

$q  \equiv 1\pmod{4}$ &  $ n|(q - 1), n < q - 1$ &   [\ref{Yan}] \\ \hline
$q \equiv 1\pmod{4}$ &  $4^{n}\cdot n^{2} \leq q$ &  [\ref{JX}]\\ \hline

  $q = p^k $, odd prime $p$ & $n= p^r + 1$, $r|k$ &   [\ref{Yan}] \\ \hline
$q = p^k $, odd prime $p$ & $n= 2p^e$, $1 \leq e < k$, $\eta(-1) = 1$&  [\ref{Yan}] \\ \hline
\end{longtable}
 \end{center}
\newpage

\begin{center}
\begin{longtable}{|c|c|c|}  
\caption{Our  results} \\\hline $q$ &   $n$ (even)  &  Reference
\\\hline
$q=r^2$, $r$ odd & $n=tm$, $2\leq t \leq \frac{r-1}{\gcd(r-1,m)}$, even $\frac{q-1}{m}$ & Theorem \ref{mt even and n=mt} \\ \hline
$q=r^2$, $r$ odd & $n=tm+1$, $tm$ odd, $2\leq t \leq \frac{r-1}{\gcd(r-1,m)}$ and $m|(q-1)$  & Theorem \ref{mt odd and n=mt+1}\\ \hline
$q=r^2$, $r$ odd & $n=tm+2$, $tm$ even, $2\leq t \leq \frac{r-1}{\gcd(r-1,m)}$ and $m|(q-1)$   &  Theorem \ref{mt even and n=mt+2}\\\hline
$q=p^m$, odd prime $p$ & $n= 2tp^e$, $2t \mid (p-1)$ and $e<m$, even $\frac{q-1}{2t}$ & Theorem \ref{hayet}\\
\hline
\end{longtable}
 \end{center}

\section{Generalized Reed-Solomon codes}
In this section, we introduce some basic notations and results on generalized
Reed-Solomon code. Throughout this paper, let $\mathbb{F}_{q}$ be a finite field with $q$ elements, and let $n$ be a positive integer
with $1<n<q$. Choose $\mathbf{a} = (\alpha_{1}, \ldots, \alpha_{n})$ to be an $n$-tuple of distinct elements of $\mathbb{F}_{q}$. Put
$\mathbf{v}=(v_{1}, \ldots, v_{n})$ with $v_i\in \F_q^*$.  For an integer $k$ with $0\leq k\leq n$, then linear code
\begin{equation}\label{def GRS}
\mathbf{GRS}_{k}(\mathbf{a}, \mathbf{v}) = \{(v_{1}f(\alpha_1), \ldots, v_{n}f(\alpha_n)): f(x)\in {\mathbb{F}_{q}}[x], \deg(f(x))\leq k-1\}
\end{equation}
is the generalized Reed-Solomon or $\mathbf{GRS}$ code.\\
The code $\mathbf{GRS}_{k}( \mathbf{a}, \mathbf{v})$ has a generator matrix \\
\[\mathbf{G}_{k}(\mathbf{a}, \mathbf{v})=\begin{pmatrix}
    v_1 & v_2 & \dots  & v_n \\
   v_1\alpha_1 & v_2\alpha_2 &  \dots  & v_n\alpha_n\\
    \vdots & \vdots  & \ddots & \vdots \\
    v_1\alpha_{1}^{k-1} & v_2\alpha_{2}^{k-1} &  \dots  & v_n\alpha_{n}^{k-1}\end{pmatrix}.\]

It is well-known that the code $\mathbf{GRS}_{k}(\mathbf{a}, \mathbf{v})$  is a $q$-ary $[n, k]$-MDS code and its dual is also MDS.

We define
\begin{equation}\label{dfn L}
L_\mathbf{a}(\alpha_{i})= {\prod_{1 \leq j \leq n, j\neq i}{(\alpha_{i}-\alpha_{j})}}.
\end{equation}

The dual of $\mathbf{GRS}$ code is explicitly determined. Precisely, let $\mathbf{1}$ be the all-one vector with appropriate length.
The dual of $\mathbf{GRS}_{k}(\mathbf{a}, \mathbf{1})$ is $\mathbf{GRS}_{n-k}(\mathbf{a}, \mathbf{u})$,
where $\mathbf{u}=(u_{1}, \ldots, u_{n})$ with $u_{i} = L_\mathbf{a}(\alpha_{i})^{-1}$ for $1\leq i \leq n$.

\begin{remark}\label{a}
 The matrix $\mathbf{G}_{k}(\mathbf{a}, \mathbf{v})$ is a generator matrix of
 $\mathbf{GRS}_{k}(\mathbf{a}, \mathbf{v})$, and $\mathbf{G}_{k}(\mathbf{a}, \mathbf{1})$ is
 a generator matrix of $\mathbf{GRS}_{k}(\mathbf{a}, \mathbf{1})$. It is clear that
 $$\mathbf{G}_{k}(\mathbf{a}, \mathbf{v}) = \mathbf{G}_{k}(\mathbf{a}, \mathbf{1})diag(v_1, v_2, \dots, v_n), $$
 where $diag(v_1, v_2, \dots, v_n)$ is the diagonal matrix with diagonal entries $v_1, v_2, \dots, v_n$.
 \end{remark}

  We now introduce some basic notations and results on extended generalized Reed-Solomon code.
The $k$-dimensional extended $\mathbf{GRS}$ code of length $n$ given by
\begin{equation}\label{def GRS infty}
\mathbf{GRS}_{k}(\mathbf{a}, \mathbf{v}, \infty)=\{(v_{1}f(\alpha_1), \ldots, v_{n-1}f(\alpha_{n-1}), f_{k-1}):
f(x)\in {\mathbb{F}_{q}}[x], \deg(f(x))\leq k-1\},
\end{equation}
where $f_{k-1}$ is the coefficient of $x^{k-1}$ in $f(x)$.\\
Clearly, the code $\mathbf{GRS}_{k}(\mathbf{a}, \mathbf{v}, \infty)$ has a generator matrix
\[\mathbf{G}_{k}(\mathbf{a}, \mathbf{v}, \infty)=\begin{pmatrix}
    v_1 & v_2 & \dots  & v_{n-1} & 0 \\
   v_1\alpha_1 & v_2\alpha_2 &  \dots  & v_{n-1}\alpha_{n-1} & 0 \\
    \vdots & \vdots  & \ddots & \vdots & \vdots   \\
    v_1\alpha_{1}^{k-1} & v_2\alpha_{2}^{k-1} &  \dots  & v_{n-1}\alpha_{n-1}^{k-1} & 1 \end{pmatrix}.\]
It is known that $\mathbf{GRS}_{k}(\mathbf{a}, \mathbf{v}, \infty)$ is $q$-ary $[n, k]$-MDS code  and its dual is also MDS.
Precisely, the dual code of $\mathbf{GRS}_{k}(\mathbf{a}, \mathbf{1}, \infty)$ is $\mathbf{GRS}_{n-k}(\mathbf{a}, \mathbf{u}, \infty)$,
where $\mathbf{u}=(u_{1}, \ldots, u_{n})$ with $u_{i} = L_\mathbf{a}(\alpha_{i})^{-1}$ for $1\leq i \leq n-1$.
\begin{remark}\label{12}
 Let $\mathbf{G}_{k}(\mathbf{a}, \mathbf{v}, \infty)$ be a generator matrix for $\mathbf{GRS}_{k}(\mathbf{a}, \mathbf{v}, \infty)$,
 and let $\mathbf{G}_{k}(\mathbf{a}, \mathbf{1}, \infty)$ be a generator matrix for $\mathbf{GRS}_{k}(\mathbf{a}, \mathbf{1}, \infty)$.
 It is easy to see that the relationship between $\mathbf{G}_{k}(\mathbf{a}, \mathbf{v}, \infty)$, and $\mathbf{G}_{k}(\mathbf{a}, \mathbf{1}, \infty)$
 is follows
 \[\mathbf{G}_{k}(\mathbf{a}, \mathbf{v}, \infty)=\mathbf{G}_{k}(\mathbf{a}, \mathbf{1}, \infty)diag(v_1, v_2, \dots, v_{n-1}, 1), \]
 where $diag(v_1, v_2, \dots, v_{n-1}, 1)$ is the diagonal matrix with diagonal entries $v_1, v_2, \dots, v_{n-1}, 1$.
 \end{remark}

 The following lemma gives a criterion for a $\mathbf{GRS}$ code to be self-dual.
 \begin{lemma}\label{y}
Let $n$ be an even integer, and $k=\frac{n}{2}$. The code $\mathbf{GRS}_{k}(\mathbf{a}, \mathbf{v})$ is a $q$-ary $[n, k]$ MDS
self-dual code over $\mathbb{F}_{q}$ if and only if for any codeword $\mathbf{c} \in \mathbf{GRS}_{k}(\mathbf{a}, \mathbf{v})$, where
$\mathbf{c}=(v_{1}f(\alpha_1), \ldots, v_{n}f(\alpha_n))$, we have
$(v_{1}^{2}f(\alpha_1), \ldots, v_{n}^{2}f(\alpha_n))\in \mathbf{GRS}_{k}(\mathbf{a}, \mathbf{1})^{\bot}$.
   \end{lemma}
\begin{proof}
 Since $\mathbf{GRS}_{k}(\mathbf{a}, \mathbf{v})=\mathbf{GRS}_{k}(\mathbf{a}, \mathbf{v})^\bot$, for any $\mathbf{c}=(v_{1}f(\alpha_1), \ldots, v_{n}f(\alpha_n))\in \mathbf{GRS}_{k}(\mathbf{a},
 \mathbf{v})$,  we
 have
 $\mathbf{G}_{k}(\mathbf{a}, \mathbf{v}){\mathbf{c}}^{T}=0$, where $\mathbf{G}_{k}(\mathbf{a}, \mathbf{v})$ is a generator matrix for
 $\mathbf{GRS}_{k}(\mathbf{a}, \mathbf{v})$. By remark(\ref{a}), we can obtain that
 \[
 \begin{array}{rcl}
 \mathbf{G}_{k}(\mathbf{a}, \mathbf{v})\mathbf{c}^{T} &=&
 (\mathbf{G}_{k}(\mathbf{a}, \mathbf{1})\cdot diag(v_1, v_2, \dots, v_n))\cdot \mathbf{c}^{T} \\[2mm]
 &=&\mathbf{G}_{k}(\mathbf{a},
 \mathbf{1})\cdot (v_{1}^{2}f(\alpha_1), \ldots, v_{n}^{2}f(\alpha_n))^{T}\\[2mm]
 &=&0.
 \end{array}
 \]
This implies that $\mathbf{GRS}_{k}(\mathbf{a}, \mathbf{v})$ is MDS self-dual if and only if
\[\mathbf{G}_{k}(\mathbf{a}, \mathbf{1})\cdot (v_{1}^{2}f(\alpha_1), \ldots, v_{n}^{2}f(\alpha_n))^{T}=0.\] In other words,
$(v_{1}^{2}f(\alpha_1), \ldots, v_{n}^{2}f(\alpha_n))\in \mathbf{GRS}_{k}(\mathbf{a}, \mathbf{1})^{\bot}$. Therefore, the lemma is proved.
\end{proof}
\begin{corollary}\label{c1}([\ref{JX}], Corollary 2.4) For  an even integer $n$, and $k=\frac{n}{2}$. If there exist
$\lambda \in \mathbb{F}_q^* $ such that $\lambda
L_\mathbf{a}(\alpha_{i})= w_{i}^{2}$ for some $w_{i} \in
\mathbb{F}_{q}^{\ast}$ for all $1\leq i \leq n$, then the code
$\mathbf{GRS}_{k}(\mathbf{a}, \mathbf{v})$ defined in (\ref{def
GRS}) is MDS self-dual, where $v_{i}= w_{i}^{-1}$ for all $1\leq i
\leq n$. \qquad $\square$
\end{corollary}
The following lemma is important and it gives the necessary and sufficient condition
to construct self-dual codes via the extended $\mathbf{GRS}$ codes.
\begin{lemma}\label{b}
Let $n$  be even and $k=\frac{n}{2}$.
 The code $\mathbf{GRS}_{k}(\mathbf{a}, \mathbf{v}, \infty)$ is a $q$-ary $[n, k]$ MDS self-dual code over $\mathbb{F}_{q}$
  if and only if for any codeword $\mathbf{c} \in \mathbf{GRS}_{k}(\mathbf{a}, \mathbf{v}, \infty)$,
  where $\mathbf{c}=(v_{1}f(\alpha_1), \ldots, v_{n-1}f(\alpha_{n-1}), f_{k-1} )$,
  we have $(v_{1}^{2}f(\alpha_1), \ldots, v_{n-1}^{2}f(\alpha_{n-1}), f_{k-1})\in \mathbf{GRS}_{k}(\mathbf{a}, \mathbf{1}, \infty)^{\bot}$.
   \end{lemma}
 \begin{proof}
  Since $\mathbf{GRS}_{k}(\mathbf{a}, \mathbf{v}, \infty)$ is a $q$-ary $[n, k]$ MDS self-dual code , we have
  $\mathbf{G}_{k}(\mathbf{a}, \mathbf{v}, \infty){\mathbf{c}}^T=0$, where $\mathbf{G}_{k}(\mathbf{a}, \mathbf{v}, \infty)$ is a
  generator matrix for $\mathbf{GRS}_{k}(\mathbf{a}, \mathbf{v}, \infty)$, with $\mathbf{c}\in \mathbf{GRS}_{k}(\mathbf{a}, \mathbf{v}, \infty)$.
   By remark(\ref{12}), we get
  \[
  \begin{array}{rcl}
 \mathbf{G}_{k}(\mathbf{a}, \mathbf{v}, \infty)\mathbf{c}^{T}&=&(\mathbf{G}_{k}(\mathbf{a}, \mathbf{1}, \infty)diag(v_1, v_2, \dots, v_{n-1}, 1))\cdot\mathbf{c}^{T}\\[2mm]
 &=&
\mathbf{G}_{k}(\mathbf{a}, \mathbf{v}, \infty)\cdot (v_{1}^{2}f(\alpha_1), \ldots, v_{n-1}^{2}f(\alpha_{n-1}), f_{k-1} )^{T}\\[2mm]
 &=&0.
 \end{array}
 \]
 It follows that $\mathbf{GRS}_{k}(\mathbf{a}, \mathbf{v}, \infty)$ is a $q$-ary $[n+1, k]$ MDS self-dual code over $\mathbb{F}_{q}$
 if and only if \[\mathbf{G}_{k}(\mathbf{a}, \mathbf{1}, \infty)(v_{1}^{2}f(\alpha_1), \ldots, v_{n-1}^{2}f(\alpha_{n-1}), f_{k-1})^{T}=0.\]
 That means $(v_{1}^{2}f(\alpha_1), \ldots, v_{n-1}^{2}f(\alpha_{n-1}), f_{k-1})\in \mathbf{GRS}_{k}(\mathbf{a}, \mathbf{1}, \infty)^{\bot}$.
  This completes the proof.

 \end{proof}

\begin{corollary}\label{c2}([\ref{Yan}], Lemma 2) Let $n$ be even and $k=\frac{n}{2}$.
If $-L_\mathbf{a}(\alpha_{i})= w_{i}^{2}$ for some $w_{i}  \in
\mathbb{F}_{q}^{\ast}$ for all $1\leq i \leq n-1$, then the code
$\mathbf{GRS}_{k}(\mathbf{a}, \mathbf{v}, \infty)$ defined in
(\ref{def GRS infty}) is MDS self-dual, where $v_{i}=w_{i}^{-1}$ for
all $1\leq i \leq n-1$. \qquad $\square$
\end{corollary}

 \section{Main Results}

 The existence of MDS self-dual codes over  finite fields of even characteristic has been completely addressed in [\ref{GA}].
In this section, we construct serval new classes of MDS self-dual codes,
by using $\mathbf{GRS}$ codes and extended $\mathbf{GRS}$ codes over finite fields of odd characteristic.\\
The following lemma can be found in [\ref{Yan}].
 \begin{lemma}\label{ha}
 Let $m|(q-1)$ be a positive integer and let $\alpha \in \mathbb{F}_{q}$ be a primitive $m$-th root of unity.
 Then for any $1\leq i \leq m$, we have
 $$ {\prod_{1 \leq j \leq m, j\neq i}{(\alpha^{i}-\alpha^{j})}}=m\alpha^{-i}.\qquad \qquad\qquad\qquad\square$$
\end{lemma}

Let $\mathbb{F}_{q}^{\ast 2}$ denote the set of nonzero squares of $\mathbb{F}_{q}$.

\begin{theorem}\label{mt even and n=mt}
Let $q=r^2$ where $r$ is odd prime power,
and let $n=tm$ be a  positive integer with $2\leq t\leq
\frac{r-1}{\gcd(r-1,m)}$ and $m\mid (q-1)$. Assume both
 $\frac{q-1}{m}$ and $n$ are even.   Then there exists a $q$-ary $[n, \frac{n}{2}]$ MDS self-dual code over
   $\mathbb{F}_{q}$.
\end{theorem}

\begin{proof}

Let $\alpha$ be a primitive $m$-th root of unity. Since one has
group isomorphism and embedding
\[\mathbb{F}_r^*\Big{/}\left(\mathbb{F}_r^*\cap \langle\alpha\rangle\right)\simeq \left(\mathbb{F}_r^*\times \langle\alpha\rangle\right)\Big{/}
\langle\alpha\rangle\leq \mathbb{F}_q^*\Big{/} \langle\alpha\rangle,\] and
$t\leq \frac{r-1}{\gcd(r-1,m)}$,  we can choose $\beta_{i}\in
\mathbb{F}_{r}^{\ast}$ ($0\leq i\leq t-1$) to be cost representatives of
$\left(\mathbb{F}_r^*\times \langle\alpha\rangle\right)\Big{/}\langle\alpha\rangle$ and
 $$\mathbf{a} =
(\alpha\beta_{0}, \ldots, \alpha^{m}\beta_{0}, \alpha\beta_{1},
\ldots, \alpha^{m}\beta_{1}, \ldots, \alpha\beta_{t-1}, \ldots,
\alpha^{m}\beta_{t-1} ).$$
 Then the entries of $\mathbf{a}$ are distinct elements of
$\mathbb{F}_{q}^*$. Note that
 \[x^m-\gamma^m=\prod_{1 \leq j \leq m}{(x-\gamma\alpha^{j})}\] for any $\gamma\in \mathbb{F}_{q}$. By Lemma \ref{ha}, for any $1\leq i \leq m$
 and for any $0\leq z \leq t-1$, we have
\[
\begin{array}{rcl}
L_{\mathbf{a}}(\beta_{z}\alpha^{i})
 &=&
 {\prod_{1 \leq j \leq m, j\neq i}{(\beta_{z}\alpha^{i}-\beta_{z}\alpha^{j})}}
 \times \prod_{l=0, l\neq z}^{t-1}{\prod_{j=1}^m{(\beta_{z}\alpha^{i}-\beta_{l}\alpha^{j})}}\\[2mm]
 &=&{\beta_{z}}^{m-1}\cdot
 m\cdot \alpha^{-i}\cdot \prod_{l=0, l\neq

 z}^{t-1}{({\beta_{z}}^{m}-{\beta_{l}}^{m})}.\\[2mm]
\end{array}
 \]
Note that $ \beta_{z}^{m-1},
 \prod_{l=0, l\neq z}^{t-1}{({\beta_{z}}^{m}-{\beta_{l}}^{m})}, m \in \mathbb{F}_r^* \subset \mathbb{F}_{q}^{*2}$ since  $\mathbb{F}_{r^2} = \mathbb{F}_{q}$.

  Since $\frac{q-1}{m}$ is even and $\alpha$ is primitive $m$-th root of unity,
  then $\alpha \in \mathbb{F}_{q}^{*2}$. As a consequence, $L_{\mathbf{a}}(\beta_{z}\alpha^{i})\in \mathbb{F}_{q}^{*2}$.
  By Corollary \ref{c1}, there exists a
   $q$-ary $[n, \frac{n}{2}]$ MDS self-dual code over $\mathbb{F}_{q}$.\\

\end{proof}

\begin{remark}
Theorem 3.4 (i) in [\ref{JX}] is a special case of the preceding result when $m=1$.
\end{remark}

\begin{example}When $r=9$ and $q=81$, we can choose $m=5$ and $t=6$. In this case $n=30$ and the preceding result shows that there exists $q$-ary MDS self-dual code of length $30$ which, as our best knowledge, has not been found in previous references.
\end{example}

\begin{theorem}\label{mt odd and n=mt+1}
Let $q=r^2$ where $r$ is odd prime power,
and let $n=tm+1$ be a positive integer with $2\leq t\leq
\frac{r-1}{\gcd(r-1,m)}$ and $m|(q-1)$. Assume $tm$ is odd. Then there exists a $q$-ary $[n, \frac{n}{2}]$ MDS self-dual code
over $\mathbb{F}_{q}$.
\end{theorem}

\begin{proof}
The proof is similar as that of Theorem \ref{mt even and n=mt}.
\[
\begin{array}{rcl}
L_{\mathbf{a}}(\beta_{z}\alpha^{i})
 &=&{\beta_{z}}^{m-1}\cdot
 m\cdot \alpha^{-i}\cdot \prod_{l=0, l\neq
 z}^{t-1}{({\beta_{z}}^{m}-{\beta_{l}}^{m})}.\\[2mm]
\end{array}
 \]
Note that $ \beta_{z}^{m-1},
 \prod_{l=0, l\neq z}^{t-1}{({\beta_{z}}^{m}-{\beta_{l}}^{m})},-m \in \mathbb{F}_r^* \subset \mathbb{F}_{q}^{*2}$ since $\mathbb{F}_{r^2} = \mathbb{F}_{q}$.

  Since $m$ is odd, then $\alpha \in \mathbb{F}_{q}^{*2}$. Therefore
 $-L_{\mathbf{a}}(\beta_{z}\alpha^{i})\in \F_q^{*2}$. By Corollary \ref{c2}, there exists a
   $q$-ary $[n, \frac{n}{2}]$ MDS self-dual code over $\mathbb{F}_{q}$.

\end{proof}

\begin{theorem}\label{mt even and n=mt+2}
Let $q=r^2$  where $r$ is odd prime power,
and let $n=tm+2$ be a positive integer with $2\leq t\leq
\frac{r-1}{\gcd(r-1,m)}$ and $m|(q-1)$. Assume $tm$ is even. Then there exists a $q$-ary $[n, \frac{n}{2}]$
MDS self-dual code over
   $\mathbb{F}_{q}$.
\end{theorem}
\begin{proof}
The elements $\alpha$ and $\beta_i$ are chosen in the same way as in
Theorem \ref{mt even and n=mt}. We define the generalized
Reed-Solomn code $\mathbf{GRS}_{k}(\mathbf{a}, \mathbf{v}, \infty)$
with
\[\mathbf{a} =
(0, \alpha\beta_{0}, \ldots, \alpha^{m}\beta_{0}, \alpha\beta_{1}, \ldots, \alpha^{m}\beta_{1}, \ldots, \alpha\beta_{t-1}, \ldots,
\alpha^{m}\beta_{t-1} ).\] Similar as the above proofs, for any
$1\leq i \leq m$
 and for any $0\leq z \leq t-1$,
\[
\begin{array}{rcl}
L_{\mathbf{a}}(\beta_{z}\alpha^{i})
 &=&\beta_z \alpha^i\times
 {\prod_{1 \leq j \leq m, j\neq i}{(\beta_{z}\alpha^{i}-\beta_{z}\alpha^{j})}}
 \times \prod_{l=0, l\neq z}^{t-1}{\prod_{1 \leq j \leq m}{(\beta_{z}\alpha^{i}-\beta_{l}\alpha^{j})}}\\[2mm]
 &=&{\beta_{z}}^{m}\cdot
 m\cdot  \prod_{l=0, l\neq
 z}^{t-1}{({\beta_{z}}^{m}-{\beta_{l}}^{m})}\\[2mm]
\end{array}
 \]
 and
 \[
 L_{\mathbf{a}}(0)
 =
 \prod_{l=0}^{t-1}{\prod_{j=1}^m{(0-\beta_{l}\alpha^{j})}}=\left(\prod_{l=0}^{t-1}\beta_l\right)^m
 \cdot\alpha^{\frac{m(m+1)}{2}}=\pm
 \left(\prod_{l=0}^{t-1}\beta_l\right)^m.
 \]
Note that $\beta_{l}, {\beta_{z}}^{m}-{\beta_{l}}^{m}$, $-1$ and $m \in \mathbb{F}_{q}^{*2}$. Hence both $-L_{\mathbf{a}}(\beta_{z}\alpha^{i})$ and $-L_{\mathbf{a}}(0)$
are in $\F_q^{*2}$. By Corollary \ref{c2}, there exists a
$q$-ary MDS self-dual code of length $n=mt+2$.
\end{proof}

 \begin{theorem}\label{hayet}
 Let $q=p^m$ with $p$ odd prime.
 For any $t$ with $2t \mid (p-1)$ and $e<m$, if $\frac{q-1}{2t}$ is even, then there exists self-dual MDS code with length $2tp^e$.
 \end{theorem}
 \begin{proof}
 Let $V$ be an $e$-dimensional $\mathbb{F}_{p}$-vector subspace in $\mathbb{F}_{q}$ satisfying $V \cap\ \mathbb{F}_{p}=0$.
 Denote by $ \omega\in \mathbb{F}_{p}$ a primitive element of order $2t$.  Choose\\
 $$\mathbf{a}= \bigcup\limits_{j=0}^{2t-1}\left(\omega^j+V\right).$$
 For any $b \in \omega^i+V$,
 \begin{equation}\label{def L}
 \begin{array}{rcl}
 L_\mathbf{a}(b)&=
 &\left(\prod\limits_{0\neq v\in V} v\right)\cdot \left(\prod\limits_{j=0, j\neq i}^{2t-1}\prod\limits_{v\in V}\left(\omega^i-\omega^j+v\right)\right)\\[4mm]
 &=&\left(\prod\limits_{0\neq v\in V} v\right)\cdot\left(\prod\limits_{v\in V} \omega^{i(2t-1)} \prod\limits_{h=1}^{2t-1} \left(1+ \omega^{-i}v-\omega^h\right)\right)\\[4mm]
 &=& \omega^{-i} \left(\prod\limits_{0\neq v\in V} v\right)\cdot \left(\prod\limits_{v\in V} \prod\limits_{h=1}^{2t-1} \left(1+ v-\omega^h\right)\right)\\[4mm]
\end{array}
 \end{equation}
 where the last equality follows from that $\prod\limits_{v\in V}\omega^{i(2t-1)}=\omega^{-ip^e}=\omega^{-i}$ and $\omega^{-i}v$ runs through $V$
 when $v$ runs through $V$.

 Denote by $c=\left(\prod\limits_{0\neq v\in V} v\right)\cdot \left(\prod\limits_{v\in V} \prod\limits_{h=1}^{2t-1} \left(1+ v-\omega^h\right)\right)$.
 Then
 $ L_\mathbf{a}(b)=\omega^{-i} c$. Since $\frac{q-1}{2t}$ is even,
 we can deduce that $\omega \in \mathbb{F}_{q}^{*2}$. As a consequence,
 $\eta\left(L_\mathbf{a}(b)\right)=\eta(c)$ which is independent of the choice of $b$. In this case we choose $\lambda= c$. Then
 $\eta\left(\lambda L_\mathbf{a}(b)\right)=\eta(\lambda c)=1$. By Corollary \ref{c1},
 there exists self-dual MDS code with length $|\mathbf{a}|=2tp^e$.
 \end{proof}
\begin{remark}
For $t=1$, the preceding result is exactly Theorem 4 (ii) in [\ref{Yan}].
\end{remark}

In the case $q$ is a square, usually Theorems 1-4 will present more
classes of MDS self-dual codes than the previous results.
\begin{example}
For $q=151^2$,
  there are  243 different $n$ for which MDS
  self-dual code of lengths $n$ are constructed in all the previous
  works(in Table 1).
  In our constructions (Theorems 1-4), there are 713 different classes of MDS self-dual codes of different lengths.
\end{example}

\section*{Acknowledgements}
{This work is supported by the self-determined research funds of
CCNU from the colleges' basic research and operation of MOE (Grant
No. CCNU18TS028). The work of J.Luo is also supported by NSFC under
Grant 11471008.}

\begin{thebibliography}{99}
\item \label{AKL} T. Aaron Gulliver, J.L. Kim, and Y. Lee, New MDS and near-
MDS self-dual codes, \emph{IEEE Trans. Inf. Theory}, vol. 54, no. 9, pp. 4354--4360, Sept. 2008.

\item \label{BGG} K. Betsumiya, S. Georgiou, T. A. Gulliver, M. Harada, and
 C. Koukouvinos, On self-dual codes over some prime fields, \emph{Discrete
 Math.}, vol. 262, nos. 1-3, pp. 37--58, 2003.

\item \label{CDG} R. Cramer, V. Daza, I. Gracia, J. J. Urroz, G. Leander, J. Marti-Farre, C. Padro, On codes, matroids and secure multi-party computation
from linear secret sharing schemes, \emph{IEEE Trans. Inf. Theory},
vol. 54, no. 6, pp. 2647--2657, June 2008.

\item \label{DMS} S. T. Dougherty, S. Mesnager, and P. Sol$\acute{\mathrm{e}}$, Secret-sharing schemes
based on self-dual codes, in \emph{Proc. Inf. Theory Workshop},
pp. 338--342, May 2008.

\item\label{DSDY} S.H. Dau, W. Song, Z. Dong, and C. Yuen,  ``Balanced spareseset generator matrices for MDS codes,
" in \emph{Proc. Inter. Symp. Inf. Theory}, July 2013, pp.
1889--1893.

\item \label{GK} S. Georgiou and C. Koukouvinos, MDS self-dual codes over large
prime fields, \emph{Finite Fields Their Appl.}, vol. 8, no. 4, pp. 455--470,
Oct. 2002.

\item\label{GA} M. Grassl and T. Aaron Gulliver, On self-dual MDS codes, \emph{Proceedings of ISIT}, 2008, pp. 1954--1957.

\item\label{GUE} K. Guenda, New MDS self-dual codes over finite fields,\emph{Designs Codes Cryptogr.}, vol. 62, no. 1, pp. 31--42, Jan. 2012.

\item \label{JX} L. Jin and C. Xing, New MDS self-dual codes from generalized Reed-
Solomon codes, \emph{IEEE Trans. Inf. Theory}, vol. 63, no. 3, pp. 1434--1438, Mar. 2017.

\item \label{KKO} J.I. Kokkala, D.S. Krotov and P.R.J.
$\ddot{\mathrm{O}}$sterg$\ddot{\mathrm{a}}$rd, ``On the
classification of MDS codes", \emph{IEEE Trans. Inf. Theory}, vol.
61, no. 2, pp. 6485--6492, Dec. 2015.

\item \label{LR} E. Louidor and R.M. Roth ``Lowest density MDS codes over extension
alphabets", \emph{IEEE Trans. Inf. Theory}, vol. 52, no. 2, pp.
3187--3197, July 2006.

\item \label{MS} F. J. MacWilliams and N. J. A. Sloane, \emph{The Theory of Error-Correcting Codes}. Amsterdam, The Netherlands: North Holland, 1977.

\item \label{MAS} J. Massey, Some applications of coding theory in cryptography, in
\emph{Proc. 4th IMA Conf. Cryptogr. Coding,} 1995, pp. 33--47.

\item \label{KL} J. L. Kim and Y. Lee, Euclidean and Hermitian self-dual MDS codes over large finite fields, \emph{J. Combinat. Theory, Series A}, vol. 105,
 no. 1, pp. 79--95, Jan. 2004.

\item\label{Yan} H. Yan, A note on the construction of MDS self-dual codes, \emph{Cryptogr. Commun.}, Published online, March 2018, see https://doi.org/10.1007/s12095-018-0288-3

\end {thebibliography}
\end{document}